%% file: main.tex
\documentclass[sigconf,nonacm]{aamas} 

\usepackage{balance} 

\usepackage[bookmarksnumbered,unicode]{hyperref}
\usepackage{amsthm}
\usepackage{todonotes}
\usepackage{caption}
\usepackage{subcaption}
\usepackage{tikz}
\usepackage{svg}
\usepackage{multirow}
\usepackage{aligned-overset}

\usepackage[ruled]{algorithm2e}

\acmSubmissionID{79}

\include{preamble}

\title{Equilibrium Learning in Combinatorial Auctions:\\Computing Approximate Bayesian Nash Equilibria via Pseudogradient Dynamics}

  \author{%
  Stefan Heidekrüger}
 \affiliation{\institution{Technical University of Munich}}
  \email{stefan.heidekrueger@in.tum.de}

  \author{%
  Paul Sutterer}
 \affiliation{\institution{Technical University of Munich}}
  \email{paul.sutterer@in.tum.de}

  \author{%
  Nils Kohring}
 \affiliation{\institution{Technical University of Munich}}
  \email{kohring@in.tum.de}

  \author{%
  Maximilian Fichtl}
 \affiliation{\institution{Technical University of Munich}}
  \email{max.fichtl@in.tum.de}

  \author{%
  Martin Bichler}
 \affiliation{\institution{Technical University of Munich}}
  \email{bichler@in.tum.de}

\begin{abstract}
  Applications of combinatorial auctions (CA) as market mechanisms are prevalent in practice, yet their Bayesian Nash equilibria (BNE) remain poorly understood. Analytical solutions are known only for a few cases where the problem can be reformulated as a tractable partial differential equation (PDE). In the general case, finding BNE is known to be computationally hard. Previous work on numerical computation of BNE in auctions has relied either on solving such PDEs explicitly, calculating pointwise best-responses in strategy space, or iteratively solving restricted subgames. In this study, we present a generic yet scalable alternative multi-agent equilibrium learning method that represents strategies as neural networks and applies policy iteration based on gradient dynamics in self-play. Most auctions are ex-post non\-differentiable, so gradients may be unavailable or misleading, and we rely on suitable pseudogradient estimates instead. Although it is well-known that gradient dynamics cannot guarantee convergence to NE in general, we observe fast and robust convergence to approximate BNE in a wide variety of auctions and present a sufficient condition for convergence.
\end{abstract}


\newcommand{\BibTeX}{\rm B\kern-.05em{\sc i\kern-.025em b}\kern-.08em\TeX}

\begin{document}


\pagestyle{fancy}
\fancyhead{}


\maketitle 

\section{Introduction}

Auctions are widely used in advertising, procurement, or for spectrum sales \citep{bichler2017HandbookSpectrumAuction,milgrom2017DiscoveringPricesAuction,ashlagi2011SimultaneousAdAuctions}. 
Auction markets inherently involve incomplete information about competitors and strategic behavior of market participants. Understanding decision making in such markets has long been an important line of research in game theory. Auctions are typically modeled as Bayesian games and one is particularly interested in the equilibria of such games.

It is well-known that equilibrium computation is hard: Finding Nash equilibria is known to be PPAD-complete even for normal-form games, which assume complete information and finite action spaces, and where a Nash equilibrium is guaranteed to exist \citep{daskalakis2009ComplexityComputingNash}. In auction games modeled as Bayesian games with continuous type and action spaces, agents' values are drawn from some continuous prior value distribution and their strategies are described as continuous bid functions on these valuations. For markets of a single item, the landmark results by \citet{vickrey1961CounterspeculationAuctionsCompetitive} have enabled a deep understanding of common auction formats. 
For multi-item auctions and more specifically \emph{combinatorial auctions}, in which players bid on \emph{bundles} of multiple items simultaneously, there has been little progress. While the complexity of computing Bayes-Nash equilibria (BNE) is not well understood, \citet{cai2014SimultaneousBayesianAuctions} show that BNE computation for a specific combinatorial auction is already (at least) PP-hard. Furthermore, finding an $\epsilon$-approximation to a BNE is still NP-hard. Explicit solutions exist for very few specific environments, but in general, we neither know whether a BNE exists nor do we have a solution theory. 
Combinatorial auctions have become a pivotal research problem in algorithmic game theory \citep{roughgarden2016TwentyLecturesAlgorithmic} and they are widely used in the field \citep{bichler2017HandbookSpectrumAuction,cramton2004CombinatorialAuctions}. Thus, understanding their equilibria is paramount, and access to scalable numerical methods for computing or approximating BNE can have a significant impact.

Equilibrium learning in games differs from most learning tasks in that it suffers from the \emph{nonstationarity problem}: Each player's objective depends on other agents' actions. Prior literature on equilibrium learning primarily focuses on complete-information games. In contrast, we focus on Bayes-Nash equilibria in games with continuous action space and continuous prior type distributions. The literature on equilibrium computation for these games is in its infancy and largely relies on best-response computations. 

In this paper, we propose Neural Pseudogradient Ascent (NPGA) as a equilibrium learning method that follows gradient dynamics. While learning based on gradient dynamics has been used in complete-information games, this is not the case for Bayesian auction games: First, the underlying problem is equivalent to an infinite-dimensional variational inequality, for which we do not know an exact solution method. Second, the ex-post payoff function of auction games is non-differentiable. Finally, multi-agent gradient dynamics are known to converge to Nash equilibria only in restricted classes of games, even under complete information.

NPGA relies on self-play with neural networks, uses evolutionary strategies to compute gradients, and can exploit GPU hardware acceleration to massively parallelize the computations. In contrast to some previous work on numerical BNE computation, NPGA does not require any setting-specific subprocedures or information beyond evaluating auction outcomes themselves, and it can thus be applied to arbitrary Bayesian games. We discuss a sufficient condition for convergence of NPGA to a unique Bayes-Nash equilibrium and provide extensive experimental results on single-item and combinatorial auctions, which pose a benchmark problem in algorithmic game theory. Interestingly, we observe convergence of NPGA to approximate BNE in a wide range of small- and medium-sized combinatorial auction environments and recover the analytical Bayes-Nash equilibrium whenever it is known.

The remainder of this paper is structured as follows: In Section \ref{sec:problem-statement}, we formally introduce the model and the problem. Section \ref{sec:related} examines related work, both in the problem domain of combinatorial auctions and related to our methodology. Next, we introduce and discuss NPGA in Section \ref{sec:our-method}, before applying it to a suite of previously studied combinatorial auctions in Section \ref{sec:results}. Finally, we summarize our findings and outline future research directions.

\section{Problem statement}\label{sec:problem-statement}

\paragraph{Bayesian Games and Combinatorial Auctions.}

A \emph{Bayesian game} or \emph{incomplete information game} is a quintuple $G = (\mathcal I,\mathcal A, \mathcal V,F,u)$. $\mathcal I = \{1, \dots, n\}$ describes the set of agents participating in the game. 
$\mathcal A \equiv \mathcal A_1 \times {\cdots} \times \mathcal A_{n}$ is the set of possible action profiles, with $\mathcal A_i$ being the set of actions available to agent $i {\in} \mathcal I$. 
$\mathcal V \equiv {\mathcal V_1 \times {\cdots} \times \mathcal V_{n}}$ is the set of \emph{type profiles}.
$F{:}\ \mathcal V \rightarrow [0,1]$ defines a joint prior probability distribution over type profiles that is assumed to be common knowledge among all agents. For any dependent random variable $X$, we denote its cumulative distribution function by  $F_X$ and its probability density function by $f_X$. For example, $F_{v_i}$ denotes the marginal distribution of $i$'s type. 
At the beginning of the game, nature draws a type profile $v {\sim} F$ and each agent $i$ is informed of their own type $v_i\in \mathcal V_i$ only, thus the type constitutes private information based on which each agent chooses their action $b_i \in \mathcal A_i$. 
Each agent's \textit{ex-post} utility function is then determined by $u_i: \mathcal A \times \mathcal V_i \rightarrow \mathbb{R}$, i.e.\ the agent's utility depends on all agents' actions but only on their own type. Agents aim to maximize their individual utility or \emph{payoff} $u_i$. Throughout this paper, we denote by the index $-i$ a profile of types, actions or strategies for all agents but agent $i$. 

In this paper, we consider \emph{sealed-bid combinatorial auctions} (CA) on $\mathcal M = \{1, \dots, m\}$ items. In such an auction, each agent, or \emph{bidder}, is allocated a bundle $k_i \in \mathcal K \equiv 2^{\mathcal M}$ of items (possibly $k_i {=} \emptyset$). Each agent's types $v_i\in \mathcal V_i$ are given by a vector of \emph{private valuations} over bundles, i.e.\ $v_i \equiv (v_i(k))_{k\in\mathcal K}$. Bidders then submit actions, called \emph{bids} $b_i$, according to some bid-language: In the general case, where bidders might be interested in any combination of items, bids are in $\mathcal A_i \subseteq \R^{|\mathcal K|}_+$, i.e. each player must submit $2^m$ bids. In practice this is prohibitive, and one commonly studies settings where valuations exhibit some structure that allows reducing the dimensionality of both the type and strategy spaces. The settings we study in Section \ref{sec:results} have type and action spaces $\R_+$ or $\R_+^2$. 

After observing their own type $v_i$, bidders submit bids ${b_i {=} \beta_i(v_i)}$ chosen according to some \emph{strategy} or \emph{bid function} ${\beta_i: \mathcal V_i \rightarrow \mathcal A_i}$ that maps individual valuations to a probability distribution over possible actions.\footnote{\emph{Mixed} strategies that randomize over actions would also be possible, but we restrict ourselves to \emph{pure} or \emph{deterministic} strategies that choose a specific action with certainty, as most work in auction theory focuses on pure-strategy Bayesian Nash equilibria.} We denote by $\Sigma_i \subseteq \mathcal A_i^{\mathcal V_i}$ the resulting strategy space of bidder $i$ and by $\Sigma \equiv \prod_i \Sigma_i$ the space of possible joint strategies. Note that even for deterministic strategies, the spaces $\Sigma_i$ are infinite-\emph{dimensional} unless $\mathcal V_i$ are finite.
The auctioneer collects these bids, applies some \emph{auction mechanism} that determines (a) an allocation $x \in \mathcal K^n$; each bidder $i$ receives a (possibly empty) bundle $x_i {\in} \mathcal K$, s.t.\ the union of these bundles is disjoint, $\dot{\bigcup}_i x_i \subseteq \mathcal K$, i.e.\ each item $m\in \mathcal M$ is allocated to at most one bidder, and (b) payments $p\in \R^n$ that the agents have to pay to the auctioneer. For brevity, we will restrict ourselves to bidders with \emph{quasi-linear} utility functions given by $u_i: \mathcal V_i \times \mathcal A \rightarrow \R$,
\begin{equation}\label{eq:utility}
    u_i(v_i, b_i, b_{-i}) = v_i (x_i) - p_i,
\end{equation}
i.e.\ the utility of each player is given by how much she values the goods that she is allocated minus the price she has to pay for them.\footnote{Quasi-linear utilities correspond to risk-neutral bidders, but our method is also applicable to bidders with risk-averse utility functions, e.g. $u=\sqrt{v-p}$.} Throughout this paper, we will differentiate between the \emph{ex-ante} state of the game, where players know only the priors $F$, the \emph{ex-interim} state, where players additionally know their own valuation $v_i \sim F_{v_i}$, and the \emph{ex-post} state, where all actions have been played and $u_i(v,b)$ can be observed.

\paragraph{Equilibria in Bayesian games.}

In non-cooperative game theory, Nash equilibria (NE) are the central equilibrium solution concept. A set of bids $b^*$ is a pure-strategy NE of the complete-information game $G=(\mathcal{I}, \mathcal A, u)$ if $u_i(b_i^*, b_{-i}^*) \geq u_i(b_i, b_{-i}^*)$ for all $b_i \in \mathcal A_i$ and all $i \in \mathcal I$. In a NE no agent has an incentive to deviate unilaterally, given the equilibrium strategy of all other agents. Bayesian-Nash equilibria (BNE) extend this notion to incomplete-information games, calculating the expected utility $\overline{u}$ over the conditional distribution of opponent valuations $v_{-i}$.
For a valuation $v_i \in \mathcal V_i$, action $b_i \in \mathcal A_i$ and fixed opponent strategies $\beta_{-i} \in \Sigma_{-i}$, we denote the \emph{ex-interim utility} of bidder $i$ by 
\begin{equation}\label{eq:iterim-util}
    \overline{u}_i(v_i, b_i, \beta_{-i}) \equiv \expect[v_{-i} \vert v_i]{u_i\left(v_i, b_i , \beta_{-i}(v_{-i})\right)}.
\end{equation}
We also denote the \emph{ex-interim utility loss} of action $b_i$ incurred by not playing the best response action, given $v_i$ and $\beta_{-i}$, by
\begin{equation}\label{eq:ex-interim-loss}
    \overline \ell_i(b_i; v_i, \beta_{-i}) = \sup_{b'_i \in \mathcal A_i} \overline u_i(v_i, b'_i, \beta_{-i}) - \overline u_i(v_i, b_i, \beta_{-i}).
\end{equation}
Note that $\overline \ell_i$ can generally not be observed in online-settings because it requires knowledge of a best-response.

An ex-interim \emph{$\epsilon$-Bayes Nash Equilibrium ($\epsilon$-BNE)} is a strategy profile $\beta^* = (\beta^*_1, \dots, \beta^*_n) \in \Sigma$ such that no agent can improve her own ex-interim expected utility by more than $\epsilon \geq 0$ by deviating from the common strategy profile. Thus, in an $\epsilon$-BNE, we have:
\begin{align}\label{eq:BNE}
    \overline{\ell}_i\left(b_i; v_i, \beta^*_{-i}\right)  \leq  \epsilon \quad \textnormal{for all } i\in \mathcal I,  v_i \in \mathcal V_i \textnormal{ and } b_i \in \mathcal A_i.
\end{align}

A $0$-BNE is simply called BNE. Thus, in a BNE, every bidder's strategy maximizes her expected ex-interim utility given opponent strategies everywhere on her type space $\mathcal V_i$. While BNE are often defined at the \emph{ex-interim} stage of the game, we also consider \emph{ex-ante} Bayesian equilibria as strategy profiles that concurrently maximize each player's \emph{ex-ante} expected utility $\tilde u$. We analogously define $\tilde u$ and the \emph{ex-ante utility losses} $\tilde \ell$ of a strategy profile $\beta \in \Sigma$ by 
\begin{align}
	\tilde{u}_i(\beta_i, \beta_{-i}) &\equiv \expect[v]{u_i(v_i, \beta_i(v_i), \beta_{-i}(v_{-i}))}\\
	\overset{\text{eq. \ref{eq:iterim-util}}}&{=} \expect[v_i \sim F_{v_i}]{\overline{u}_i(v_i, b_i, \beta_{-i})}
\end{align}
and
\begin{align}
    \label{eq:ex-ante-loss} \tilde \ell_i(\beta_i, \beta_{-i}) &\equiv \sup_{\beta'_i \in \Sigma_i} \tilde u_i(\beta'_i, \beta_{-i}) - \tilde u_i(\beta_i, \beta_{-i}).
\end{align}
Then, an ex-ante BNE $\beta^* \in \Sigma$ can be characterized by the equations $\tilde \ell_i(\beta^*_i, \beta^*_{-i}) = 0$ for all $i \in \mathcal I.$
Clearly, every ex-interim BNE also constitutes an ex-ante equilibrium. The reverse holds almost surely, i.e.\ any ex-ante equilibrium fulfills Equation \ref{eq:BNE}, except possibly on a nullset $V{\subset}\mathcal V$, i.e. with $\int_V df_v(v) = 0$. To see this, one may consider the equation $0 = \tilde \ell (\beta^*) = \expect[v_i]{\overline \ell(\beta^*, v_i)}$ and the fact that $\overline \ell(\beta, v_i) \geq 0$ by definition. In this paper, we concern ourselves with finding ex-ante equilibria of auction games.

\begin{figure}
	\centering
	\includegraphics*[trim=0 21 0 5,clip,width=0.3\textwidth]{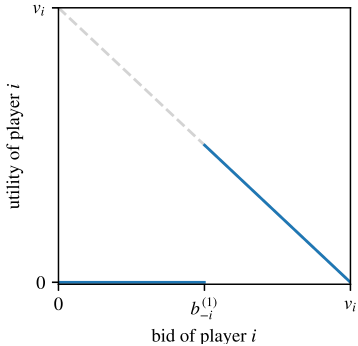}
	\caption{Discontinuous ex-post utility function $u_i(b_i)$ in single-item First-Price Sealed-Bid auction for fixed opponent bids $b_{-i}$ with highest opponent bid $b^{(1)}_{-i}$.}\label{fig:degenerate_gradient}
\end{figure}

\section{Related work}\label{sec:related}

\emph{Gradient dynamics in games} have been studied in evolutionary game theory and multiagent learning. While earlier work considered mixed strategies over normal-form games \citep{zinkevich2003OnlineConvexProgramming,bowling2002MultiagentLearningUsing,bowling2005ConvergenceNoRegretMultiagent,busoniu2008ComprehensiveSurveyMultiagent}, more recently, motivated by the emergence of GANs, there has been a focus on (complete-information) games with continuous action spaces and smooth utility functions \citep{mertikopoulos2019LearningGamesContinuous,letcher2019DifferentiableGameMechanics,balduzzi2018MechanicsNPlayerDifferentiable,schaefer2019CompetitiveGradientDescent}. A common result for many settings and algorithms is that gradient-based learning rules do not necessarily converge to Nash equilibria and may exhibit cycling behavior, but often achieve no-regret properties and thus converge to weaker Coarse Correlated equilibria (CCE). An analogous result exists for finite-type Bayesian games, where no-regret learners are guaranteed to converge to a Bayesian CCE \citep{hartline2015NoRegretLearningBayesian}. In the present paper, we study equilibrium learning via gradient dynamics in \emph{continuous-type} Bayesian games, specifically auctions, where they have not been investigated previously to our knowledge.\\

\emph{Equilibrium computation in auctions.} Earlier approaches to find equilibria in auctions were usually setting specific and relied on reformulating Equation \ref{eq:BNE} as a differential equation (where possible), then solving this equation analytically or numerically \citep{vickrey1961CounterspeculationAuctionsCompetitive,krishna2009AuctionTheory,ausubel2019CoreselectingAuctionsIncomplete}.  \citet{armantier2008ApproximationNashEquilibria} introduced a BNE-computation method that is based on expressing the Bayesian game as the limit of a sequence complete-information games. They show that the sequence of Nash equilibria in the restricted games converges to a BNE of the original game. While this result holds for any Bayesian game, setting-specific information is required to generate and solve the restricted games. \citet{rabinovich2013ComputingPureBayesianNash} study best-response dynamics on mixed strategies in auctions with finite action spaces. Most recently, \citet{bosshard2017ComputingBayesNashEquilibria,bosshard2020ComputingBayesNashEquilibria} proposed a method to find BNE in combinatorial auctions that relies on smoothed best-response dynamics and is applicable to any Bayesian game. The method explicitly computes point-wise best-responses in a fine-grained linearization of the strategy space via sophisticated Monte-Carlo integration. Our method, NPGA does likewise not require setting-specific information while also avoiding costly explicit best-response computations.

Some recent work has explored other topics at the intersection of deep learning and game theory. \cite{schuurmans2016DeepLearningGames} reformulate supervised NN training into finding NE in a corresponding game. \cite{dutting2017OptimalAuctionsDeep} use deep learning to design revenue-optimal and truthful auction mechanisms from the auctioneer's perspective. In contrast, we explore how to bid optimally from the perspective of an auction participant.

\section{Pseudogradient dynamics in auction games}\label{sec:our-method}

Next, we present our method for equilibrium computation in auctions, which we call Neural Pseudogradient Ascent (NPGA).
On a high level, we propose following the gradient dynamics of the game via simultaneous gradient ascent of all bidders. As we will see, however, computing the gradients themselves is not straightforward in the auction setting and we will need some modifications to established gradient dynamics methods such as \cite{zinkevich2003OnlineConvexProgramming,silver2014DeterministicPolicyGradient}. For now, assume that players observe a gradient-oracle $\nabla_{\beta_i} \tilde u_i(\beta_i, \beta_{-i})$ with respect to the current strategy profile $\beta^t$ in each iteration. Then the rule proposes that players perform a projected gradient update:
\begin{align}\label{eq:gradient-step}
  \beta_i^t \equiv \mathcal P_{\Sigma_i}\left(\beta_i^{t-1} + \Delta_i^t  \right ) \quad \text{ with } \quad \Delta_i^t \propto \nabla_{\beta_i} \tilde u_i(\beta_i, \beta_{-i}),
\end{align}
where $\mathcal P_{\Sigma_i}(\argdot)$ is the projection onto the set of feasible strategies for agent $i$. Several things must be noted about Equation \ref{eq:gradient-step}: First, we consider the gradient dynamics of the \emph{ex-ante} utility $\tilde u$, rather than ex-interim or ex-post utilities. The goal of an individual update step is thus to marginally improve the expected utility of player $i$ across all possible joint valuations $v \sim F$. This perspective ultimately considers low-probability events less important than high-probability events, which is in contrast to some other methods, which explicitly aim to optimize \emph{all} ex-interim states \citep{bosshard2017ComputingBayesNashEquilibria}. Second, to compute the gradient oracle $\nabla_{\beta} \tilde u$ in self-play, we rely on access to other players strategies, but evaluating each player's policy relies only on their own valuation. We thus follow the centralized-training, decentralized-execution framework common in multi-agent learning. Third, $\beta_i \in \Sigma_i$ are functions in an infinite-dimensional function space, so the gradient $\nabla_{\beta_i} \tilde u_i$ is itself a \emph{functional} derivative. In our ex-ante perspective, we thus consider this to be the Gateaux derivative over the Hilbert space $\Sigma_i$, equipped with the inner product $\langle \psi, \beta_i \rangle = \expect[v_i\sim F_{v_i}]{\psi(v_i)^T \beta_i(v_i) }$ (which, in turn, defines the projection in Equation \ref{eq:gradient-step} as $\mathcal P_{\Sigma_i}(\beta) \equiv \arg\min_{\sigma\in\Sigma_i} \langle \sigma-\beta,\sigma-\beta \rangle$).

\paragraph{Policy networks.} To implement this derivative in practice, we represent each bidder's strategy by a \emph{policy network}
$\beta_i(v_i) \equiv \pi_i(v_i; \theta_i)$
specified by a neural network architecture and a corresponding parameter vector $\theta_i \in \R^{d_i}$. Importantly, given a suitable network architecture, one can ensure that all $\theta_i$ always yield feasible bids, thus making the projection in the update step obsolete. In the empirical part of this study, we restrict ourselves to fully-connected feed-forward neural networks with ReLU activations in the output layer, which ensure nonnegative bids---the only feasibility constraint in the auctions we study. In any case, $d_i \in \N$ is finite and we thus transform the problem of choosing an infinite-dimensional strategy into choosing a finite-dimensional parameter vector $\theta_i$.

\paragraph{Policy pseudogradients.} The \emph{deterministic policy gradient theorem} \citep{silver2014DeterministicPolicyGradient} gives an established, canonical way to compute the payoff gradient with respect to the parameters $\theta$: $\nabla_{\theta_i} \tilde u_i(\pi_i(\argdot; \theta_i), \beta_{-i}) = {\expect[v\sim F]{\nabla_{\theta_i} \pi(v_i; \theta_i) \nabla_{b_i}u_i(v_i; b_i, \beta_{-i}(v_{-i}))\vert_{b_i = \pi_i(v_i; \theta_i)}}}$. 
However, the regularity conditions required by the theorem are commonly violated in combinatorial auctions. In particular, due to the discrete nature of the allocations $x$, the ex-post utilities $u_i(v_i,b_i,b_{-i})$ are usually discontinuous---and thus not (sub)differentiable in $b_i$. While this nondifferentiability does not extend to $\tilde u$, it nevertheless renders the policy gradient formula above inapplicable: Although the set of discontinuities is a $v$-nullset in practice, its true gradient provides systematically misleading signals, even on the differentiable intervals of $u_i(v_i, {\argdot}, b_{-i})$: Consider a first-price sealed-bid auction in which winning bidders pay their bid amount $b_i$. The utility graph is separated into two sections (see Figure \ref{fig:degenerate_gradient}): (a) Bidding lower than the highest opposing bid leads to zero payoff and thus no learning feedback, $\nabla_{b_i} {u}_i {=} 0$, and (b) winning with a the highest bid \emph{must} yields feedback to decrease the bid, $\nabla_{b_i} {u}_i {=} -1$. Back-propagation will thus lead to a steady decrease of bids, until all players bid constant zero at any valuation.

To alleviate this, we instead estimate the policy gradient using a finite difference approach based on evolutionary strategies (ES) \cite{salimans2017EvolutionStrategiesScalable}. To calculate $\nabla_\theta \tilde u$, we perturb the parameter vector $P$ times, $\theta_{i;p} \equiv \theta_i + \varepsilon_p$, using zero-mean Gaussian noise $\varepsilon_p {\sim} \mathcal N(0, \sigma^2)$ for $p\in \{1,\dots,P\}$, where $P, \sigma$ are hyperparameters. We then calculate each perturbation's \emph{fitness}, $\varphi_p \equiv \tilde{u}_i(\pi_i(v_i; \theta_{i;p}), \beta_{-i})$, via Monte-Carlo integration, and estimate the gradients as the fitness-weighted perturbation noise $\nabla_\theta^{ES} \equiv \frac{1}{\sigma^2 P}\sum_{p} \varphi_p\varepsilon_p$. \citet{salimans2017EvolutionStrategiesScalable} motivated their application of this ES gradient estimate to reinforcement learning because it's applicable to parallelization across large-scale CPU-clusters, but here we instead exploit its property that it gives an asymptotically unbiased estimator of $\nabla_\theta \tilde u$ even when $\nabla_b u$ itself is not well-defined. To summarize, NPGA ``implements'' Equation \eqref{eq:gradient-step} via ES-pseudogradients and a neural network parametrization of strategy functions which renders the projection step unnecessary:
\begin{align}\label{eq:npga-step}
  \beta_i^t \equiv \pi_i(\argdot; \theta_i^t) \quad \text{with} \quad \theta_i^t \equiv \theta_i^{t-1} + \Delta_i^t  \quad \text{where} \quad \Delta_i^t \propto \nabla_{\theta_i^t}^{ES}.
\end{align}

\paragraph{Vectorizing auction evaluations.} The only information about the game $G$ needed in the computation of this learning rule is the evaluation of $\tilde u = \expect[v\sim F]{u}$ for a given strategy profile. Given a vectorized implementation of the joint ex-post utility $u$, estimating $\tilde u$ via Monte-Carlo integration over $\mathcal V$ is suitable to parallel execution on hardware accelerators such as GPUs. To this end, we built custom vectorized implementations of many common auction mechanisms using the PyTorch framework \citep{paszke2017AutomaticDifferentiationPyTorch}, enabling Monte-Carlo estimation multiple orders of magnitude faster compared to previous numerical work on auctions. For moderately sized auction games commonly studied in the literature, allocations $x$ can be computed in a vectorized fashion via full enumeration of feasible allocations.\footnote{We stress that, with current hardware, this approach remains intractable for larger auctions that are applied in the real world, e.g. spectrum auctions.} Common payment rules either have inherently vectorizable closed-form formulation (e.g.\ first-price auctions) or can be reformulated as the solution of a constrained quadratic program (e.g.\ the Vickrey-Clarke-Groves (VCG) mechanism or core-selecting pricing rules \cite{day2012QuadraticCoreSelectingPayment}). To solve a large batch of the latter in parallel, we leverage a custom vectorized implementation of interior-point methods. (A similar approach has previously been used by \cite{amos2019OptNetDifferentiableOptimization}.)
\paragraph{A convergence criterion.}\label{subsec:convergence}

As discussed in Section \ref{sec:related}, gradient dynamics do not converge to Nash equilibria in general. For differentiable, finite-dimensional, complete-information games (auctions are neither!), \citet{mertikopoulos2019LearningGamesContinuous} show that strict monotonicity of the payoff gradients is a sufficient condition for almost-sure convergence of gradient dynamics to a unique Nash equilibrium as it leads to strict concavity of the game. \citet{ui2016BayesianNashEquilibrium} shows an analogous result for ex-post differentiable Bayesian games, in which payoff-monotonicity guarantees the existence of a unique BNE. However, the result likewise does not directly apply to auctions due to their ex-post nondifferentiability. Instead, we give a slightly less restrictive criterion based on ex-interim payoff monotonicity that ensures convergence of gradient dynamics and whose formulation is compatible with auction games.

\begin{definition}[Strict Ex-interim Payoff Monotonicity]\label{def:ex-i-monotonicity}
	Let $G = (\mathcal I, \mathcal A, \mathcal V, F, u)$ be a Bayesian game, such that the individual ex-interim utilities are continuously differentiable in $b_i$
	 with gradients bounded by a constant $Z{>}0$ via $\lVert \nabla_{b_i} \overline u_i(v_i, b_i, \beta_{-i})\rVert \leq Z$.
	 $G$ is called \emph{strictly (ex-interim) payoff-monotone}, if for all $i{\in} \mathcal I$,  ${\beta_{-i} {\in} \Sigma_{-i}}$, ${a_i, b_i \in \mathcal A_i}$ and almost everywhere $v_i {\in} \mathcal V_i$ the following holds:
	\begin{equation}\label{eq:monotonicity}
	  \langle \nabla_{a_i} \overline u_i(v_i, a_i, \beta_{-i}) -  \nabla_{b_i} \overline u_i(v_i, b_i, \beta_{-i}), a_i - b_i \rangle \; < \; 0.
  \end{equation}
  \end{definition}

While analytical verification of this criterion is elusive, except in special settings, it can (approximately) be checked numerically by sampling pairs of action profiles $a,b$ for all players and using finite-difference gradient-estimators. Doing so, we observe consistency with the criterion in all settings studied in Section \ref{sec:results}. In the following, we provide a convergence result for NPGA under ex-interim monotonicity. For our convergence analysis, we will rely on certain properties of ``appropriate'' neural network architectures:

\begin{definition}[Regular Convex Policy Network]\label{def:npga-policy-net}
  A \emph{Regular Convex Policy Network} is a neural network $\pi_i: \mathcal V_i \times \Theta_i \rightarrow \mathcal A_i$ with $\dim \Theta_i = d_i$ and the following properties: 
  \begin{enumerate}
	\item $\pi_i$ is a \emph{convex neural network} in its parameters: For any convex function $g{:}\ \Sigma_i {\rightarrow} \R$, the map $\theta_i \mapsto g(\pi_i(\argdot, \theta_i))$ is convex.
	\item $\pi_i$ \emph{universally approximates} $\Sigma_i$: There exists a $\delta >0$, s.t. for all $\beta_i \in \Sigma_i$ there is a parameter vector $\theta_i \in \Theta_i$ with $\expect[v_i]{\lVert \beta_i(v_i) - \pi_i(v_i, \theta_i)\rVert} \leq \delta$.
    \item \emph{Regularity:} $\pi_i$ is Lipschitz-continuous in its parameters in the sense that there's some $L{>}0$ such that for all $\theta_i, \theta_i' {\in} \Theta_i$ we have
    ${\expect[v_i]{\lVert \pi_i(v_i, \theta_i)- \pi_i(v_i, \theta'_i)\rVert} \leq L \lVert \theta_i - \theta'_i \rVert}$.
  \end{enumerate}
\end{definition}

Neural networks that are employed in practice (and in our empirical analysis) generally do not comply with Definition \ref{def:npga-policy-net}, but such networks have been shown to exist, see e.g. \citet{bach2017BreakingCurseDimensionality}, who studies wide single-hidden-layer networks with ReLU activations, in which only the output-layer weights are being trained. We'll state our main proposition before discussing this difference further:

\begin{proposition}\label{thm:proposition}
Let $G = (\mathcal I, \mathcal A, \mathcal V, F, u)$ be a Bayesian game such that the ex-post utilities exist, and such that the ex-interim payoff-gradients exist and fulfill strict ex-interim payoff monotonicity. Then, with an NN architecture as in Definition \ref{def:npga-policy-net} and appropriate update step sizes, NPGA converges to an ex-ante $\epsilon$-BNE of $G$, where $\epsilon \leq Z(2L\sigma\sqrt{d}+\delta)$.
\end{proposition}

While existence and unique\-ness of BNE in infinite-dimensional games are unknown in the general case, Proposition \ref{thm:proposition} guarantees efficient computability in a wide range of settings, some of which we explore in the next section. Still, it's important to note that there may be auctions for which payoff-monotonicity does not hold.

A proof outline Proposition \ref{thm:proposition} is given at the end of this article in Appendix \ref{app:proof}, for additional technical derivations, see the supplementary material. As demonstrated in the proof, the use of Regular Convex Policy Networks transforms the training process into a problem of finding a Nash Equilibrium in a \emph{concave}, finite-dimensional, complete-information game. Crucially, concavity of this game ensures existence of, and convergence to, a unique global equilibrium. Just as neural networks are known to find ``good'' solutions to nonconvex optimization problems in practice despite a lack of theoretical guarantees, we will see below that we observe convergence to BNE when using standard neural network architectures that don't meet Definition \ref{def:npga-policy-net}: As such, we see Regular Convex Policy Networks as a helpful tool for theoretical analysis, but their implementation is generally neither practical nor desirable while common architectures achieve similar results.

\section{Empirical Results}\label{sec:results}

We evaluate NPGA on three suites of auction theoretic settings: First, in Section \ref{subsec:single-item} we validate our method on a suite of the most commonly studied auctions, i.e. single-item auctions with symmetric priors, before considering two suites of combinatorial auctions, the LLG (Section \ref{subsec:LLG}) and LLLLGG (Section \ref{subsec:LLLLGG}) environments. In total, we study 21 different auction settings with different numbers of players, pricing rules, risk-profiles, and prior distributions of the valuations. In 18 of these settings, the (unique) BNE is known analytically, in three settings, no BNE is known.

\paragraph{Evaluation Metrics} To evaluate the quality of strategy-profiles $\beta$ learned by NPGA, we will provide four metrics, whenever available: When we have access to the analytical solution BNE $\beta^*$, we can simply check whether $\beta \rightarrow \beta^*$. To do so, we report each agent's 
\begin{enumerate}
	\setcounter{enumi}{\value{evalMetricCounter}}
	\item utility loss $\ell^*_i$ that results from unilaterally deviating from the BNE strategy profile $\beta^*$ by playing the learned strategy $\beta_i$ instead: $\ell^*_i \equiv \tilde \ell_i(\beta_i, \beta^*_{-i})$, compare Equation \ref{eq:ex-ante-loss},
	\item and the distance $\lVert \beta_i - \beta^*_i\rVert_\Sigma$ in strategy space. \setcounter{evalMetricCounter}{\value{enumi}}
\end{enumerate} 
Both of these can be estimated via Monte-Carlo integration over the valuations $v{\sim} F$, i.e. given a batch of size $H$ of valuations $(v_{h;i}, v_{h,-i})$, we approximate $\ell^*_i$ by the sample mean of $\tilde \ell_i({\beta_i(v_{h;i}),} {\beta^*_{-i}(v_{h;-i})})$ and $\lVert \beta_i - \beta^*_i\rVert$ by the RMSE of $\beta_i(v_{h;i})$ and $\beta_i^*(v_{h;i})$ in action-space. 

However, to make NPGA applicable in practice, we are also interested in judging the quality of $\beta$ when no BNE is known. To do so, we also estimate the potential gains of deviating from the current strategy profile $\hat \ell_i \approx \tilde \ell_i (\beta_i; \beta_{-i})$ as well as an estimator $\hat \epsilon$ to the ``true'' epsilon of $\beta$ (smallest $\epsilon$ s.t.\, $\beta$ forms an ex-interim $\epsilon$-BNE). In the absence of analytical solutions, one may periodically calculate these estimators and use them as a termination criterion once the desired precision is reached. As we will see, these additional metrics are expensive: To calculate the estimators $\hat \ell$ and $\hat \epsilon$ we introduce a grid of $W$ equidistant points $b_w$ per bidder, covering the action spaces $\mathcal A_i$. Now, given a valuation $v_i$ and a bid $b_i$, we approximate the ex-interim utility loss $\overline \ell(v_i, b_i, \beta_{-i})$ of $b_i$ at $v_i$ via
\begin{equation}
	\hat \lambda_i(v_i; b_i, \beta) \equiv \frac{1}{H}\max_w \sum_h u\left(v_i; b_w, \beta_{-i}(v_{h,-i})\right) - u\left(v_i; b_i, \beta_{-i}(v_{h,-i})\right).
\end{equation}
Note that the batch $H$ only runs across opponent valuations $v_{-i}$. To evaluate $\hat \lambda_i$ at a single $v_i$, we thus need ${(W{+}1) {\cdot} H}$ auction evaluations. We can then estimate the 

\begin{enumerate}
	\setcounter{enumi}{\value{evalMetricCounter}}
	\item worst-case ex-interim loss: $\hfill \hat \epsilon = \, \max_h{\hat\lambda_i(v_{h,i}; \beta_i(v_{h,i}), \beta_{-i})}$,
	\item and the ex-ante loss: $\hfill \hat \ell = \frac{1}{H} \sum_h{\hat\lambda_i(v_{h,i}; \beta_i(v_{h,i}), \beta_{-i})}$.
\end{enumerate}
Estimations of $\hat \lambda$ can be shared for both computations, nevertheless we need $\mathcal O(nWH^2)$ auction evaluations to calculate these metrics (in contrast, an iteration of NPGA requires only $\mathcal O(nPK)$ evaluations, with {$P{\ll} W$}). As the metrics $\hat \epsilon, \hat \ell$ are expensive to compute with dense grids, we use smaller batch sizes $W$ and $H$ than in evaluating $\ell^*$, and calculate these metrics only once every 100 iterations of the algorithm.

\paragraph{Hyperparameters} We use common hyperparameters across almost all settings (except where noted otherwise): Fully connected neural networks with two hidden layers of 10 nodes each with SeLU \citep{klambauer2017SelfNormalizingNeuralNetworks} activations. ES-parameters $P{=}64$, $\sigma {=} \frac{1}{\sqrt{d_i}}$. We use Adam optimizer steps with default hyperparameters as suggested in \cite{kingma2017AdamMethodStochastic}. To avoid degenerate initializations of $\theta$ (e.g. where one or more bidders bid constant zero due to dead ReLUs in the output layer), we perform supervised pre-training to the \emph{truthful strategy} $\beta_i(v_i) = v_i$. All experiments were performed on a single Nvidia GeForce 2080Ti and batch sizes in Monte-Carlo sampling were chosen to maximize GPU-RAM utilization:
A learning batch size of $K = 2^{18}$; primary evaluation batch size (for $\ell^*, ||\beta-\beta^*||$) of $H = 2^{22}$; and secondary evaluation batch size $H=2^{12}$ and grid size $W=2^{10}$ (for $\hat \ell, \hat \epsilon$). Each experiment is repeated ten times over 5{,}000 iterations each.

\begin{figure}
	\centering
	\includegraphics*[trim=0 13 0 0,clip,width=0.45\textwidth]{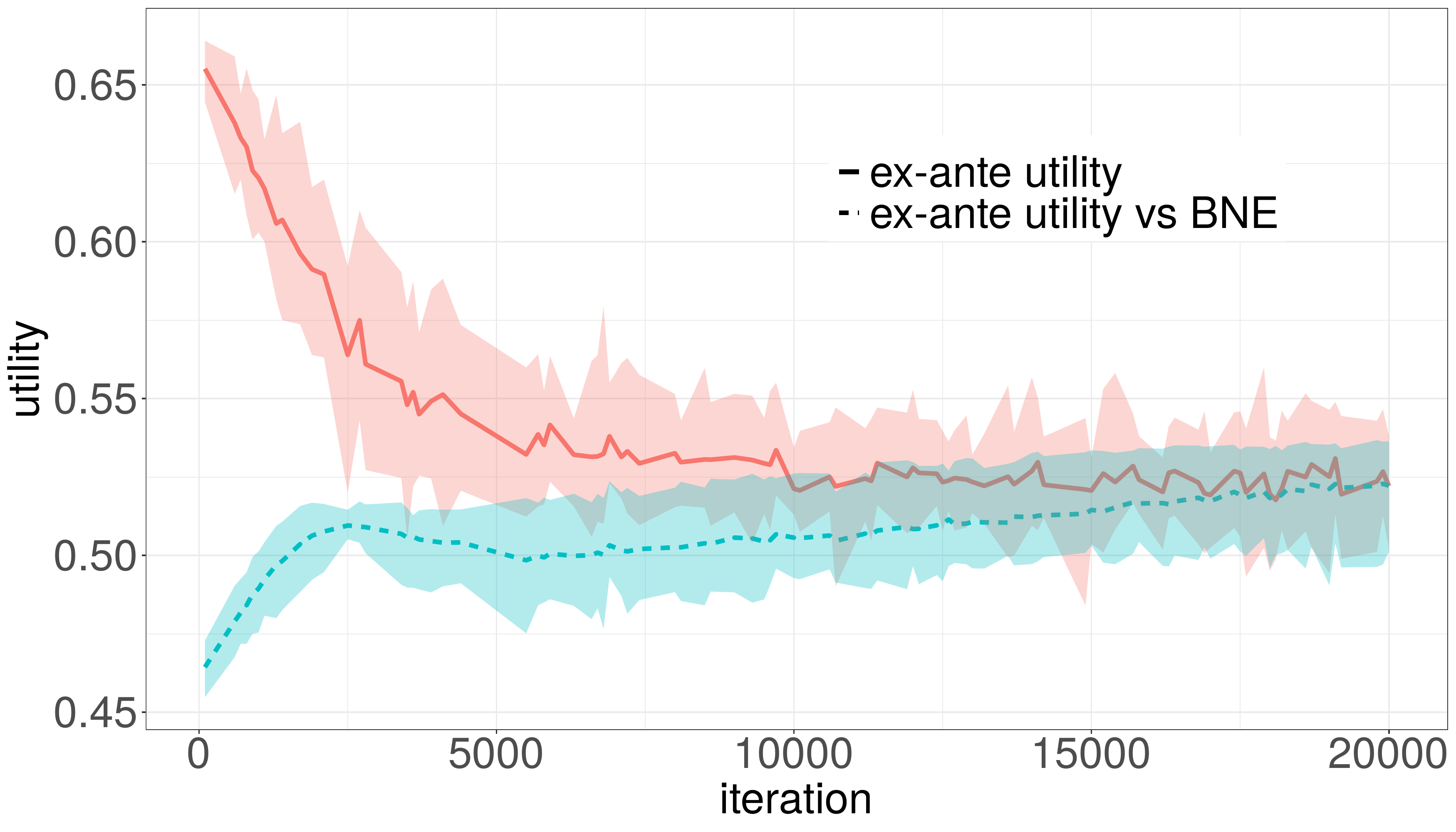}
	\caption{Learning curve of NPGA in 10-player FPSB auction with Gaussian priors, evaluated in self-play (red, solid) and against the BNE (blue, dotted). Line and shaded area indicate mean/min/max over 10 repetitions.}\label{fig:learning-curve}
\end{figure}

\subsection{Single-Item Auctions}\label{subsec:single-item}

\input{figures/auction_table_single_item.tex}

First-price sealed-bid (FPSB) auctions on a single item, in which the highest-bidding player wins the item and pays her own bid as price, are the best-known auctions and for many configurations their BNE are known analytically \cite{krishna2009AuctionTheory}. We apply NPGA to 12 such FPSB settings with 2, 3, 5 and 10 bidders with uniform- and normal-distributed valuations and risk-neutral and risk-averse 
utility functions. The results are given in Table \ref{table:single-item_table}.

In all settings, we observe convergence of NPGA to a close approximation of the analytical BNE in terms of ex-ante payoff, both when evaluated in self-play and against opponents playing the BNE. However, we also see that sometimes there's no full norm-convergence in the strategy space: This indicates that NPGA learns \emph{ex-ante} BNE as the algorithm is designed to do, but may bid suboptimally in ``unimportant'' regions of the valuation space, e.g. when there are many players and $i$'s valuation is low (see Gaussian-10p setting): Learning a highly-nuanced signal in these regions would require a larger sample-size. Similarly, the estimator $\hat\epsilon$ tends to be pessimistic, as expected. Nevertheless, we see that the strategy-space distance is low in most settings, and even when it is not, the learned strategy becomes indistinguishable from the BNE in terms of ex-ante utility: Figure \ref{fig:learning-curve} shows the learning curve for the Gaussian-10p setting where the norm has not converged. Additionally, we observe that the exploitability-estimate $\hat\ell$, while not exactly equal to $\ell^*$, is consistent in order of magnitude and may thus serve as a suitable proxy for convergence in the absence of known BNE. 

For a complete treatment of the single-item setting, we also implemented Vickrey/Second Price auctions, where NPGA consistently found the BNE (all metrics ${<}10^{{-}4}$), as well as learning using the canonical policy gradient theorem instead of pseudogradients, which, as expected (see Section 4), consistently failed to learn the BNE and instead lead to all players bidding 0 everywhere.

\subsection{Combinatorial Auctions}

\emph{Local-global combinatorial auctions} will serve as our main benchmark for BNE compuation. In such auctions, there are two groups of bidders, locals and globals: Globals $g$ are interested in larger bundles of items while their priors allow them to draw higher valuations, so local bidders $l$ need to coordinate to outbid the globals. We consider settings where locals have either independent or correlated uniform priors $v_{ik} \sim \mathcal U(0, \overline{v}_i)$ with $\overline{v}_l = 1, \overline{v}_g = 2$ (for each bundle $k\in\mathcal K_i$). The 3-player LLG setting (Section \ref{subsec:LLG}) is a standard setting in auction theory and one of the smallest CAs that requires strategic cooperation between bidders; the larger 6-player LLLLGG setting (Section \ref{subsec:LLLLGG}) was proposed by \citep{bosshard2017ComputingBayesNashEquilibria} and, to our knowledge, is the most complex environment in which approximate BNE have been computed to date.
\input{figures/auction_table_llg_new.tex}

\subsubsection{The LLG setting}\label{subsec:LLG}
The LLG setting includes two local bidders and one global bidder that bid on $m=2$ items. Local bidders $i=1,2$ are each interested in the bundle $\{i\}$, while the global bidder wants the package $\{1,2\}$ of both items. Each bidder submits a bid $b_i {\in} \R_+$ for their respective bundle. The setting has been extensively studied in the context of different \emph{core-selecting} pricing rules as commonly used in real-world spectrum auctions \citep{day2012QuadraticCoreSelectingPayment,goeree2016ImpossibilityCoreselectingAuctions}. Closed-form solutions of the unique, symmetric BNE under three such rules are known in the LLG setting for both independent and correlated priors: the nearest-VCG rule, the nearest-zero (or proxy) rule, and the nearest-bid rule. The interested reader is referred to \cite{ausubel2019CoreselectingAuctionsIncomplete} for details. In these rules, it has been shown that the global bidder is bidding truthfully in the BNE. The local bidders' BNE strategies differ in each payment rule and depending on the correlation between locals' priors. We evaluate NPGA on the three core payment rules with independent and correlated priors (correlation coefficient of local bidders: $\gamma = 0.5$) as well as the first-price payment rule with independent priors, for which no exact BNE is known.
Numerical results for all rules are presented in Table \ref{table:auction-results-llg}. We again observe that NPGA converges to the BNE in all six-settings where it is known. In fact, after low hundreds of iterations, we can no longer detect a difference in utility to the true BNE with available measurement precision, while still observing slight differences in strategy-space distance: Figure \ref{fig:llg-ngsp-learning} depicts the strategy learned by NPGA after 5{,}000 iterations in comparison to the analytical BNE strategy for the nearest-zero payment rule, and shows an almost perfect fit. In the FPSB auction, no BNE is known, but values of $\hat \ell\approx 10^{-3}$ (while the global (local) bidders achieve an ex-ante utility of $0.426$ ($0.149$)) indicate that exploitability of $\beta$ is minuscule.

\begin{figure}
	\centering

		\includegraphics*[trim=0 13 0 0,clip,width=0.45\textwidth]{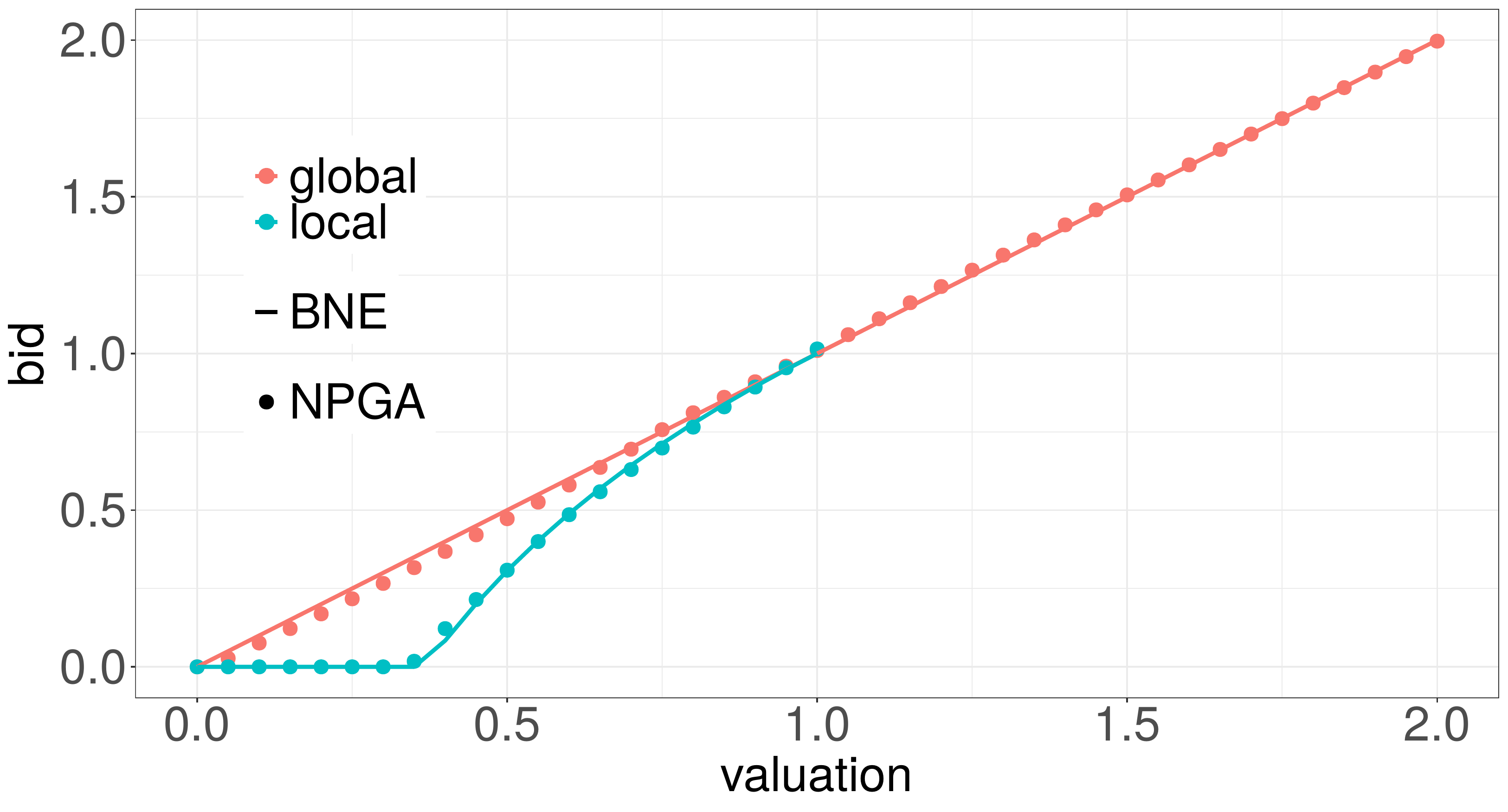} 
		\caption{Learned (dots) and BNE (lines) strategies in LLG-setting with nearest-zero core payment rule.}
		\label{fig:llg-ngsp-learning}
\end{figure}
\begin{figure}
		\centering
		\includegraphics*[trim=0 13 0 0,clip,width=0.45\textwidth]{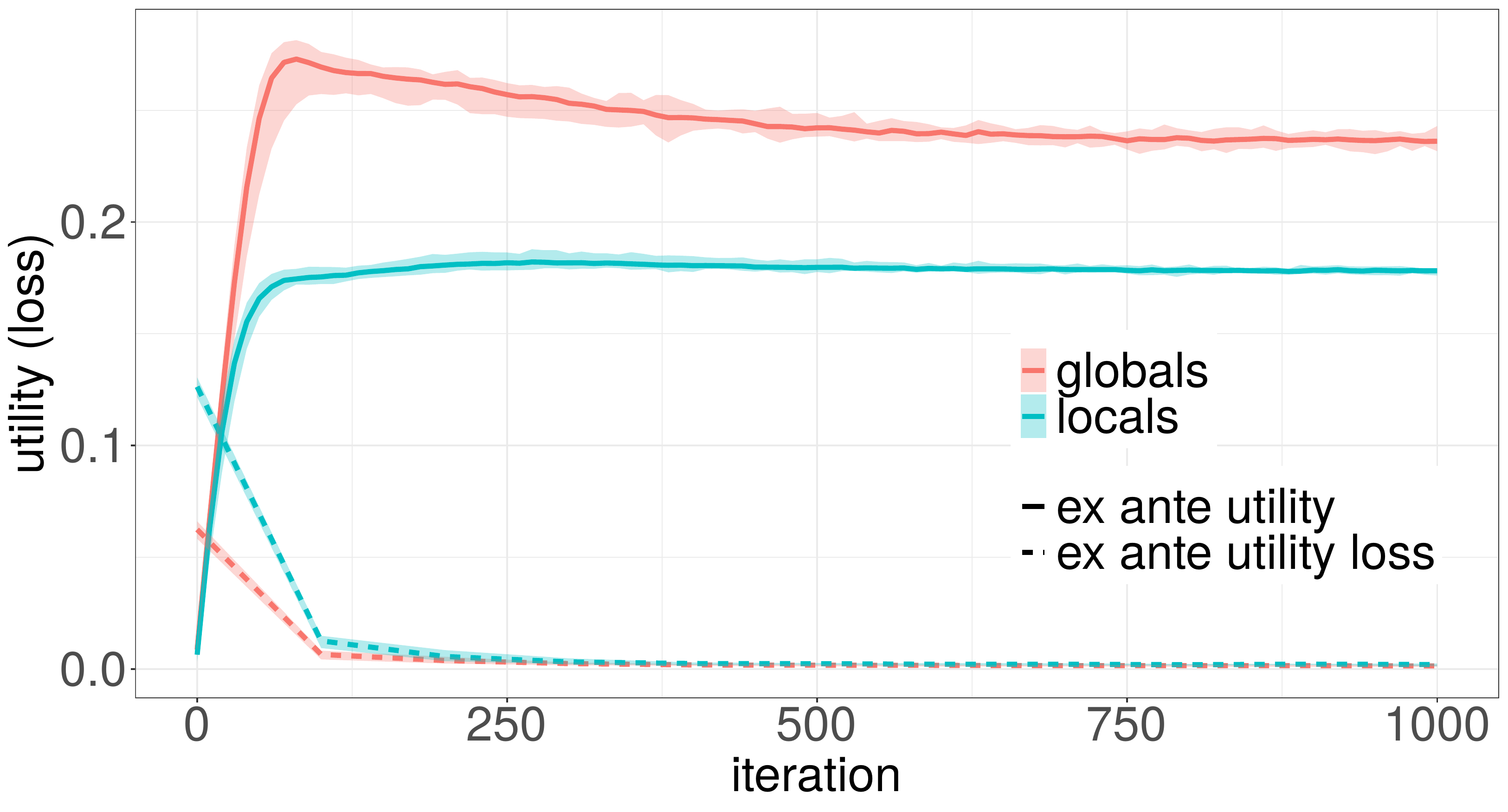}
		\caption{Ex-ante utility $\tilde u$ and estimated loss $\hat \ell$ of in NPGA self-play in the LLLLGG first-price auction. Shaded area and line show min, max, mean over 10 repetitions.}
		\label{fig:llllgg-fp}
\end{figure}

As noted above, estimating $\hat \ell$ and $\hat \epsilon$ is computationally expensive, and is not needed to run NPGA itself. The average total computation time over 5{,}000 iterations is below 70 minutes in any LLG setting.\footnote{The computation of estimated utility loss is not straightforward for correlated priors because it requires sampling from the conditional distributions $F_{v_{-i}|v_i}$ which (1) may not be available in general and (2) adds additional computational complexity. While we do not report the estimated utility loss in these setting, the more relevant utility loss (playing against BNE) is so low that we assume convergence.}

\input{figures/auction_table_llllgg_new.tex}

\subsubsection{The LLLLGG setting}\label{subsec:LLLLGG}

In the LLLLGG setting, four local and two global bidders compete for six items, where each bidder is interested in two (partly overlapping) bundles (containing 2 (local) or 4 (global) items each), with actions being represented as $\mathcal A_i = \R^2_+$. No analytical BNE are known in this setting (except for the trivial VCG pricing rule, where bidding truthfully constitutes a BNE). We apply NPGA to LLLLGG with first-price and nearest-vcg rules.

In LLLLGG, computing the clearing prices is nontrivial and computationally expensive and forms our  computational bottleneck, particularly at large batch-sizes: Nearest-VCG prices require solving a linear- and a subsequent quadratic optimization problem for each auction \citep{day2012QuadraticCoreSelectingPayment}. Because NPGA evaluates many thousand auctions in each iteration, we implemented a custom interior-point solver that can solve batches of quadratic optimization problems on the GPU. Nonetheless we make the following hyperparameter adjustments to reduce complexity in the nearest-VCG setting: $P=32$; $K = 2^{14}$; $H=2^{7}$ and $W=2^{8}$ on two experiments of 1{,}000 iterations each.

For both pricing rules, NPGA learns strategy profiles with an estimated ex-ante utility loss $\hat \ell {<} 0.002$ for both local and global bidders. Global (local) bidders achieve stable average utilities of 0.238 (0.18) in first-price and 0.181 (0.201) in nearest-vcg, thus the estimated loss indicates that players can be exploited for less than 1\% of their achieved utility. Figure \ref{fig:llllgg-fp} shows the NPGA learning-curve for both groups of bidders in the first-price setting. We see that both utility and the loss $\hat \ell$ converge fast, the latter reaching a value of $0.0005$ after just 700 iterations.

BNE-computation in the LLLLGG setting has previously been studied by \citet{bosshard2017ComputingBayesNashEquilibria,bosshard2020ComputingBayesNashEquilibria}.Their method, based on point-wise best-response dynamics, has limited comparability to NPGA, as they differ both in goals and evaluation, but one may summarize that NPGA has somewhat lower precision, being an ex-ante method and due to memory-constraints, while being significantly faster: In \cite{bosshard2017ComputingBayesNashEquilibria}, \citeauthor{bosshard2017ComputingBayesNashEquilibria} report finding an estimated ex-interim $0.0037$-BNE for the LLLLGG first-price setting in $54{,}384$ CPU-core hours; NPGA finds an estimated ex-ante $0.0015$-BNE (ex-interim $0.0109$) in 81 minutes on a single GPU ($\approx$5{,}800 CUDA-core-hours).\footnote{In a recent follow-up paper \cite{bosshard2020ComputingBayesNashEquilibria}, \citeauthor{bosshard2017ComputingBayesNashEquilibria} further improve the sample-efficiency of \cite{bosshard2017ComputingBayesNashEquilibria} using sophisticated Monte-Carlo estimation methods. Much of this optimization is equally applicable to NPGA and will be part of future work.}

\section{Conclusion and future work}

This paper explores equilibrium learning in Bayesian games, one of the large unsolved problems in algorithmic game theory. Gradient dynamics are challenging in Bayesian auction games for several reasons: these games are not differentiable, and the continuous type- and action spaces make efficient representation difficult or expensive. We propose Neural Pseudogradient Ascent as a numerical method for equilibrium learning that relies on pseudogradients on parameter spaces of policy networks. We hope that our approach will make possible the study of gradient dynamics in game-theoretical and microeconomic settings where they have previously been considered inapplicable.
In experiments, we validate NPGA on standard single-item and combinatorial auctions, which constitute a pivotal problem in algorithmic game theory with many practical applications. We find that NPGA converges to approximate BNE for central benchmark problems in this field, and we prove a sufficient criterion under which almost sure convergence to equilibria is guaranteed. In summary, the method can provide an effective numerical tool to compute approximate BNE not only for combinatorial auctions but also for other Bayesian games, without setting-specific customization, all while running on consumer hardware and leveraging GPU-parallelization for performance.

\appendix
\section{Proof of Proposition \ref{thm:proposition}\label{app:proof}}
\begin{proof}[Proof] We will approximate the infinite-dimensional Bayesian-game by a finite-dimensional (but continuous action), complete-information ``proxy game''. Under strict monotonicity, the regularity conditions above, and with Regular Convex Policy Networks, we argue that NPGA almost surely finds an approximation of the unique NE in this proxy game. We then give a bound on the ex-ante loss in the original game for this proxy-NE, thus certifying an $\epsilon$-BNE.

	First, existence of the ex-interim gradients (Definition \ref{def:ex-i-monotonicity}) implies that the ex-ante utilities
	$\tilde u_i(\beta_i, \beta_{-i}) = {\expect[v_i]{\overline u_i(v_i, \beta(v_i), \beta_{-i})}}$ are Gâteaux-differentiable in the Hilbert spaces $\Sigma_i$ with Gâteaux-gradients
	$\nabla_{\beta_i} \tilde u_i[\beta](v_i) = {\nabla_{b_i} \overline u_i(v_i, b_i, \beta_{-i})\vert_{b_i=\beta_i(v_i)}}.$ (Follows from 
	2.54, 2.55, 17.10 of \cite{bauschke2011LowerSemicontinuousConvex} and direct calculations.)
	With ex-interim payoff-monotonicity, we then have \begin{equation}
		\begin{aligned}
		 &\phantom{=}\left\langle \nabla_{\beta_i}\tilde u_i[\beta_i, \beta_{-i}] - \nabla_{\alpha_i}\tilde u_i[\alpha_i, \beta_{-i}], \beta_i - \alpha_i \right\rangle_{\Sigma_i} \\
		&= \expect[v_i]{
		  \left\langle 
			\nabla_{b_i} \overline u_i(v_i, b_i, \beta_{-i}) -\nabla_{a_i} \overline u_i(v_i, a_i, \beta_{-i})
			, \;
			 b_i{-}a_i 
		  \right\rangle \middle\vert_{\substack{b_i=\beta_i(v_i)\\a_i=\alpha_i(v_i)}}} \\
		  \overset{eq. \ref{eq:monotonicity}}&{<} \expect[v_i]{0} = 0.
	  \end{aligned}
	  \end{equation}
	  I.e.\ the ex-ante gradients are strictly monotone operators on $\Sigma_i$. It follows that $\tilde u_i$ are strictly concave in $\beta_i$ \cite[Thm 17.10]{bauschke2011LowerSemicontinuousConvex}.
	With a neural network as in Definition \ref{def:npga-policy-net}, the functions $\breve u_i(\theta_i) \equiv \tilde{u}_i(\pi_i(\argdot, \theta_i), \beta_{-i})$ are then also strictly concave in $\theta_i$ for any opponent strategies $\beta_{-i}$. We can then construct a finite-dimensional, complete-information \emph{Parameter Game} $\breve G$, in which all players approximate their strategies $\beta$ in $G$ using policy networks and we interpret the parameters $\theta\in \R^d$ of the networks as \emph{the action} of the new game: $\breve G \equiv (\mathcal I, \Theta, \breve u)$. As this game is finite-dimensional and concave, \citet{mertikopoulos2019LearningGamesContinuous} establish that (1) it has a unique Nash equilibrium $\breve \theta^*$ and (2) the \emph{dual averaging} algorithm converges almost surely to $\breve \theta^*$ given an unbiased and finite-variance oracle of the gradients $\nabla_{\theta_i}\breve u_i(\theta_i; \theta_{-i})$.  Next, we argue that $\breve \theta^*$ induces an approximate BNE in the original Bayesian Game $G$ before analyzing how NPGA implements dual averaging in $\breve G$ with noisy feedback, thus finding a good approximation of $\breve \theta^*$.
	Let $\breve \theta^*$ thus be the Nash equilibrium of $\breve G$. Then for any player $i$, $\breve \theta^*_i$ is a best response (BR) to $\breve \theta^*_{-i}$ and $\pi_i(\argdot, \breve \theta^*_i)$ is an ex-ante BR to $\pi_{-i}(\argdot, \breve \theta^*_{-i})$ in the Bayesian Game with \emph{restricted strategy space} of functions expressible by the network. As we assumed universal approximation properties of $\pi_i$, however, any BR $\beta^*_i$ in the \emph{unrestricted} game $G$ must be close in function space to $\pi_i(\argdot, \breve \theta^*_i)$, and the ex-ante utility loss incurred by not playing $\beta^*_i$ instead of $\pi_i(\argdot; \breve\theta^*_i)$ is bounded: In fact, with the Lipschitz-regularity conditions on the ex-interim gradients and universal approximability of $\pi$, we observe the following for arbitrary $\theta_{-i}$: If $\breve \theta^*\in\Theta_i$ and $\beta^*_i\in\Sigma_i$ are BRs to $\theta_{-i}$ in $\breve G$ and $G$, respectively, then
	\begin{equation}\label{eq:bound-ante-to-param}
		\begin{aligned}
	  \tilde \ell_i(\breve\theta^*; \theta_{-i})  &= \tilde u_i(\beta^*_i, \theta_{-i}) - \tilde u_i(\breve \theta^*_i, \theta_{-i}) \\
			&= \expect[v_i]{\overline{u}_i(v_i, \beta_i^*(v_i), \theta_{-i}) - \overline u_i(v_i; \breve \theta^*_i, \theta_{-i})} \\
			\overset{\textnormal{def. \ref{def:ex-i-monotonicity}}}&{\leq} Z\cdot\expect[v_i]{\lVert \beta^*_i(v_i) - \pi_i(v_i, \breve\theta^*_i)\rVert}
			\overset{\textnormal{def. \ref{def:npga-policy-net}}}{\leq} Z\delta,
		\end{aligned}
	\end{equation}
	where $\tilde u(\theta_i, \theta_{-i}) \equiv \tilde u_i(\pi_i(\argdot; \theta_i), \pi_{-i}(\argdot, \theta_{-i}))$. In the NE, all $\breve\theta^*_i$ are BRs, so we have $\tilde \ell(\breve \theta^*) \leq Z\delta$.

	Finally, we show that NPGA finds a good approximation of $\breve\theta$. As deliberated above, we choose the NN architecture in such a way that $\Theta$ becomes unconstrained, i.e.\ any parameter $\theta_i \in \R^{d_i}$ is feasible, where $d_i$ is the dimension of the network for player $i$. On an \emph{unconstrained} action set $\Theta$, however, dual averaging (with Euclidean regularization) is equivalent to Online Gradient Ascent on $\breve u$ \cite{zinkevich2003OnlineConvexProgramming,mertikopoulos2019LearningGamesContinuous}. Therefore, NPGA implements Dual Averaging on $\tilde u$ using the gradient oracle $\nabla_\theta^{ES}$. 
	
	To use the convergence result of \citet{mertikopoulos2019LearningGamesContinuous} of NPGA to $\breve \theta$, it would remain to show that the Neural Pseudogradients $\nabla^{ES}\tilde u$ are finite-variance and unbiased estimators of the true gradients $\nabla_\theta \breve u$. This is unfortunately violated for strictly positive ES-noise variance $\sigma^2$ used in NPGA (but asymptotically true for $\sigma {\rightarrow }0$). However, for $\sigma{>}0$ we can set $\breve u_i^\sigma \equiv \expect[\varepsilon\sim \mathcal N(0, \sigma^2 I)]{\breve u_i(\theta_i + \varepsilon, \theta_{-i})}$ and introduce yet another finite-dimensional game $\breve G^\sigma {=} (\mathcal I, \Theta, \breve u ^\sigma)$. Now, one can show (see supplement) that (1) $\breve G^\sigma$ is, again, concave, that (2) the ES-gradients are finite-variance, unbiased estimators of $\breve u^\sigma$, and (3) that the loss in $\breve G$ of any $\theta_i$  is bounded by that in $\breve G^\sigma$ via
	\begin{equation}\label{eq:bound-param-to-smoothed}
		\breve \ell_i(\theta_i, \theta_{-i}) \leq \breve \ell_i^\sigma(\theta_i, \theta_{-i}) + 2ZL\sqrt{d_i}\sigma.
	\end{equation}
	Due to (1), $\breve G^\sigma$ again admits a unique NE $\theta^*$ \cite[Thm 2.2]{mertikopoulos2019LearningGamesContinuous}, and with (2) and Definition \ref{def:ex-i-monotonicity}, NPGA converges to $\theta^*$ almost surely for appropriate step sizes \cite[Cor 4.8]{mertikopoulos2019LearningGamesContinuous}. 

	To summarize, we showed that NPGA finds a parameter profile $\theta^*$ that forms a NE of $\breve G^\sigma$ and which retains an-ex ante loss in $G$ of
	\begin{equation}
		\begin{split}
		\tilde \ell_i(\theta^*) &=  \tilde u_i(\beta^*_i, \theta^*_{-i}) \underbrace{- \tilde u_i(\breve \theta_i^*, \theta^*_{-i}) + \tilde u_i(\breve \theta_i^*, \theta^*_{-i})}_{=\,0} - \tilde u_i(\theta^*) \\
		&=	\tilde \ell_i(\breve \theta^*_i; \theta^*_{-i}) + \breve \ell_i(\theta^*)
		\overset{(\ref{eq:bound-ante-to-param}),(\ref{eq:bound-param-to-smoothed})}{\leq} Z\delta+2ZL\sqrt{d_i}\sigma + \underbrace{\breve \ell_i^\sigma(\theta^*)}_{\makebox[0pt]{\scriptsize{\textnormal{loss in }$\breve G^\sigma {=}0$}}}.
		\end{split}
	\end{equation}
	Thus, setting $\epsilon \equiv Z(2L\sqrt{d}\sigma + \delta)$ where $d \equiv \max_i d_i$, NPGA converges almost surely to an ex-ante $\epsilon$-BNE of $G$.
\end{proof}

\balance
\begin{acks}
We’re grateful for funding by the Deutsche Forschungsgemeinschaft (DFG, German Research Foundation) - BI 1057/1-8. We thank Vitor Bosshard, Ben Lubin, Panayotis Mertikopoulos, Sven Seuken, Takashi Ui, and Felipe Maldonado for valuable feedback, and two former students in our group: Kevin Falkenstein, for a separate implementation of initial algorithms, and Anne Christopher, for developing the custom batched QP solver. All errors are ours.
\end{acks}



\bibliographystyle{ACM-Reference-Format} 
\bibliography{bibliography}

\end{document}


\section*{Supplementary Material}
\renewcommand{\theequation}{S.\arabic{equation}}
\renewcommand{\thefigure}{S.\arabic{figure}}

\subsection*{Auxiliary Results in Proof of Proposition 1}

Let all notation be the same as in the main proof of Proposition 1 from Appendix A of the paper.

\begin{lemma}
  Let $\breve G^\sigma = (\mathcal{I}, \Theta, \breve u^\sigma)$ with $\breve u_i^\sigma \equiv \expect[\varepsilon\sim \mathcal N(0, \sigma^2 I)]{\breve u_i(\theta_i + \varepsilon, \theta_{-i})}$. Then
  \begin{enumerate}
    \item $\breve G^\sigma$ is a concave game, i.e. all $\breve u^\sigma_i$ are concave in $\theta_i$.
    \item $\nabla^{ES}\tilde u_i$ is an unbiased and finite-variance estimator of $\nabla_{\theta_i}\breve u^\sigma_i$
    \item For arbitrary $\theta_i,\theta_{-i}\in \Theta$, we have $\breve \ell_i(\theta_i, \theta_{-i}) \leq \breve \ell_i^\sigma(\theta_i, \theta_{-i}) + 2ZL\sqrt{d_i}\sigma.$
  \end{enumerate}
\end{lemma}

\begin{proof}[Proof of No. 1]

    Let $f: \R^d\rightarrow \R$ be a convex function, then $x\mapsto \expect[\varepsilon\sim \mathcal N(0,\sigma^2)]{f(x+\varepsilon)}$ is also convex: For $x,y \in \R^d$ and $\lambda \in [0,1]$ we have
    \begin{equation}\begin{aligned}
     \expect[\varepsilon \sim \mathcal{N}(0,\sigma^2)]{f(\lambda x + (1-\lambda)y+\varepsilon)}
    &= \expect[\varepsilon]{f(\lambda (x+\varepsilon) + (1-\lambda)(y+\varepsilon)} \\
    &\leq \lambda \E_{\varepsilon} [ f(x+ \varepsilon) ] + (1-\lambda)\E_{\varepsilon} [ f(y+\varepsilon)].   
    \end{aligned}\end{equation}
    Thus, as every $-\breve u_i$ is convex in $\theta_i$, so is every $-\breve u^\sigma_i$.
\end{proof}

\begin{proof}[Proof of No. 2]

For fixed $\sigma > 0$, we have
\[\breve{u}_i^{\sigma}(\theta_i,\theta_{-i}) := \E_{\varepsilon \sim \mathcal{N}(0,\sigma^2I)} [\breve{u}_i(\theta_i +  \varepsilon,\theta_{-i})].\]
This is equal to the convolution of $\breve{u}_i$ with a Gaussian kernel in the $i$-th coordinate. As noted by Salimans et al. (2017), its (exact) gradient with respect to $\theta_i$ is thus given by
\[\nabla_{\theta_i} \breve{u}_i^{\sigma}(\theta_i,\theta_{-i}) = \frac{1}{\sigma}\E_{\varepsilon \sim \mathcal{N}(0,I)}[\varepsilon(\breve{u}_i(\theta_i +  \sigma\varepsilon,\theta_{-i})- \breve{u}_i(\theta_i,\theta_{-i}))].\]
By the substitution $\varepsilon' = \sigma \varepsilon$, we see by the transformation formula that
\[\nabla_{\theta_i} \breve{u}_i^{\sigma}(\theta_i,\theta_{-i}) = \frac{1}{\sigma^2} \E_{\varepsilon \sim \mathcal{N}(0,\sigma^2I)}[\varepsilon(\breve{u}_i(\theta_i +  \varepsilon,\theta_{-i})- \breve{u}_i(\theta_i,\theta_{-i}))].\]
If we approximate this term by taking $P$ independent samples $\varepsilon_p \sim \mathcal{N}(0,\sigma^2I)$, we get
\[\nabla_{\theta_i} \breve{u}_i^{\sigma}(\theta_i,\theta_{-i}) \approx \frac{1}{P\sigma^2}\sum_p \varepsilon_p(\breve{u}_i(\theta_i + \sigma \varepsilon_p,\theta_{-i})- \breve{u}_i(\theta_i,\theta_{-i})).\]
In the same way, we can approximate $\breve{u}_i$ by sampling $H$ valuation profiles $v_h$ with respect to the distribution the valuations are drawn from:
\[\breve{u}_i(\theta_i + \sigma \varepsilon_p,\theta_{-i}) \approx \frac{1}{H} \sum_h u_i(v_{h,i},\pi_i(v_{h,i},\theta_i + \sigma\varepsilon_p),\pi_{-i}(v_{h,-i},\theta_{-i})).\]
The combination of these approximations is exactly how $\nabla^{ES}$ is computed in the NPGA Algorithm:
\[\nabla^{ES}\breve u_i(\theta_i,\theta_{-i}) = \frac{1}{PH\sigma^2}\sum_p \varepsilon_p \sum_h u_i(v_{h,i},\pi_i(v_{h,i},\theta_i + \sigma\varepsilon_p),\pi_{-i}(v_{h,-i},\theta_{-i})) - u_i(v_{h,i},\pi_i(v_{h,i},\theta_i),\pi_{-i}(v_{h,-i},\theta_{-i})).\]
 Since we sample independently and with respect to the original distributions, the approximation is in expectation equal to the true gradient. Thus, the approximation is unbiased with respect to the smoothed utilities $\breve{u}_i^{\sigma}$. $\nabla^{ES}$ also has finite mean squared error: Define
 \[ X_{p,h} = \varepsilon_p\left( u_i(v_{h,i},\pi_i(v_{h,i},\theta_i + \varepsilon_p),\pi_{-i}(v_{h,-i},\theta_{-i})) - u_i(v_{h,i},\pi_i(v_{h,i},\theta_i),\pi_{-i}(v_{h,-i},\theta_{-i})) \right). \]
 With the additional assumption (omitted in the main text) that the ex-post utilities are square-integrable (i.e. there's a $S>0$, s.t. for  all $\beta$: $\expect[v]{u(v,\beta(v))<S}$), we have
 \begin{align*}
 &\E_v[u_i(v_{h,i},\pi_i(v_{h,i},\theta_i + \varepsilon_p),\pi_{-i}(v_{h,-i},\theta_{-i}))^2] \leq S \text{ and}\\
 &\E_v[u_i(v_{h,i},\pi_i(v_{h,i},\theta_i),\pi_{-i}(v_{h,-i},\theta_{-i}))^2] \leq S.
 \end{align*}
 This implies $\E[ X_{p,h}^2 ] \leq 4S\E[\lVert \varepsilon \rVert^2] = 4Sd_i\sigma^2$, where we used the inequality $(a-b)^2 \leq 2a^2+2b^2$.
Since $\nabla^{ES} \breve u_i(\theta_i,\theta_{-i}) = \frac{1}{PH\sigma^2}\sum_{p,h} X_{p,h}$, we have that
 \begin{align*}
 &\E\left[ \nabla^{ES} \breve u_i(\theta_i,\theta_{-i})^2 \right] = \frac{1}{P^2H^2\sigma^2}\E\left[ \left(\sum_{p,h} X_{p,h}\right)^2 \right] = \frac{1}{\sigma^2}\E\left[ \left(\sum_{p,h} \frac{X_{p,h}}{PH}\right)^2 \right] \leq \\
 &\leq \frac{1}{PH\sigma^2}\E \left[ \sum_{p,h} X_{p,h}^2 \right] \leq \frac{1}{PH\sigma^2}4PHd_i\sigma^2S = 4Sd_i < \infty.
 \end{align*}
 Consequently, our gradient estimate has finite mean squared error.
\end{proof}

\begin{proof}[Proof of No. 3]
  We start by bounding the difference between the utilities of the game $\breve G$ and the game $\breve G^{\sigma}$. To be precise, we prove the following bound:
\begin{align}
|\breve u_i(\theta_i,\theta_{-i}) - \breve u_i^{\sigma}(\theta_i,\theta_{-i})| \leq ZL\sqrt{d_i}\sigma \label{eqn:sigma_bound}
\end{align}
for arbitrary strategies $\theta$. By definition, $\breve u^{\sigma}_i (\theta_i,\theta_{-i}) = \E_{\varepsilon \sim \mathcal{N}(0,\sigma^2 I)}[\breve u_i(\theta_i + \varepsilon,\theta_{-i})]$. Since $\breve u_i(\theta_i,\theta_{-i}) =  \E_{\varepsilon \sim \mathcal{N}(0,\sigma^2 I)}[\breve u_i(\theta_i,\theta_{-i})]$, we have the inequality
\begin{align}
|\breve u_i(\theta_i,\theta_{-i}) - \breve u^{\sigma}_i (\theta_i,\theta_{-i})| \leq \E_{\varepsilon \sim \mathcal{N}(0,\sigma^2 I)} [|\breve u_i(\theta_i + \varepsilon,\theta_{-i}) - \breve u_i(\theta_i,\theta_{-i})  |]. \label{eqn:exp_bound}
\end{align}
Next, we show that for fixed $\varepsilon$, $|\breve u_i(\theta_i + \varepsilon,\theta_{-i}) - \breve u_i(\theta_i,\theta_{-i})  | \leq ZL\lVert \varepsilon \rVert$. We compute
\begin{align*}
|\breve u_i(\theta_i + \varepsilon,\theta_{-i}) - \breve u_i(\theta_i,\theta_{-i})|  \leq \E_{v_i} \left[ \lvert \bar{u}_i(v_i,\pi_i(v_i,\theta_i + \varepsilon),\theta_{-i}) -\bar{u}_i(v_i,\pi_i(v_i,\theta_i),\theta_{-i}) \rvert  \right]
\end{align*}
Since by assumption, $\overline{u}_i$ is differentiable with respect to $b_i$ and the differential is uniformly bounded by $Z$ (Definition 1), we have for every $\varepsilon$
\[
\lvert \bar{u}_i(v_i,\pi_i(v_i,\theta_i + \varepsilon),\theta_{-i}) -\bar{u}_i(v_i,\pi_i(v_i,\theta_i),\theta_{-i}) \rvert \leq \left\lVert \frac{\partial \bar{u}_i}{\partial b_i} \right\rVert_{\infty} \lVert \pi_i(v_i,\theta_i + \varepsilon)-\pi_i(v_i,\theta_i) \rVert \leq Z\lVert \pi_i(v_i,\theta_i + \varepsilon)-\pi_i(v_i,\theta_i) \rVert.
\]
Consequently, by the Lipschitz-continuity in Definition 2,
\[|\breve u_i(\theta_i + \varepsilon,\theta_{-i}) - \breve u_i(\theta_i,\theta_{-i})|  \leq Z\E_{v_i} \left[  \lVert \pi_i(v_i,\theta_i + \varepsilon)-\pi_i(v_i,\theta_i) \rVert \right] \leq ZL\lVert \varepsilon \rVert \]
which implies by Equation (\ref{eqn:exp_bound})
\[|\breve u_i(\theta_i + \varepsilon,\theta_{-i}) - \breve u_i^{\sigma}(\theta_i,\theta_{-i})| \leq ZL\E_{\varepsilon \sim \mathcal{N}(0,\sigma^2 I)} [\lVert \varepsilon \rVert] \leq ZL\sqrt{d_i}\sigma.\]
This proves equation (\ref{eqn:sigma_bound}).
Now let $\breve \theta_i$ be a best response to $\theta_{-i}$ in the game $\breve G$. Then
\begin{align*}
&\breve \ell_i(\theta_i,\theta_{-i}) = \breve u_i(\theta^*_i,\theta_{-i})-\breve u_i(\theta_i,\theta_{-i}) = \\
&= \left(\breve u_i(\theta^*_i,\theta_{-i})-\breve u_i^{\sigma}(\theta^*_i,\theta_{-i})\right) +\left(\breve u_i^{\sigma}(\theta^*_i,\theta_{-i})-\breve u_i^{\sigma}(\theta_i,\theta_{-i})\right) +\left(\breve u_i^{\sigma}(\theta_i,\theta_{-i})-\breve u_i(\theta_i,\theta_{-i})\right) \leq \\
&\leq ZL\sqrt{d_i}\sigma + \breve \ell_i^{\sigma}(\theta_i,\theta_{-i}) + ZL\sqrt{d_i}\sigma = \breve \ell_i^{\sigma}(\theta_i,\theta_{-i}) + 2ZL\sqrt{d_i}\sigma.
\end{align*}
\end{proof}

%% file: preamble.tex
\newboolean{includeHidden}
\setboolean{includeHidden}{false}
\newlength{\commentWidth}
\newtheorem{definition}{Definition}

\newtheorem{proposition}{Proposition}

\DeclareMathOperator{\E}{\mathbb{E}}
\DeclareMathOperator{\R}{\mathbb{R}}
\DeclareMathOperator{\N}{\mathbb{N}}

\newcommand{\argdot}{\,\cdot\,}

\newcounter{evalMetricCounter}

\newcommand{\expect}[2][]{\E_{#1}\left[#2\right]}


%% file: figures/auction_table_single_item.tex

\begin{table}[ht]
	\caption{Performance of strategies learned by NPGA in single-item first-price sealed-bid auctions. Results are averaged over 10 runs of 5{,}000 (uniform risk-neutral) or 20{,}000 (other settings) iterations each.}
	\resizebox{0.45\textwidth}{!}{
    \begin{tabular}{crrrrrrr}
	\multirow{2}{*}{\textbf{valuations}} & 
	\multirow{2}{*}{\textbf{n}} & 
	\multirow{2}{*}{\textbf{$\ell^*$}} & 
	\multirow{2}{*}{\textbf{$\lVert\beta^*{-}\beta\rVert$}} & 
	\multirow{2}{*}{$\hat{\ell}$} & \multirow{2}{*}{$\hat{\epsilon}$}  &
	 \multicolumn{1}{l}{\multirow{2}{*}{\begin{tabular}[c]{@{}l@{}}\textbf{time}\\ sec/it\end{tabular}}}
    \\  \\\toprule
    \multirow{4}{*}{\textbf{\begin{tabular}[c]{@{}c@{}}Uniform\\ $\mathcal U(0,10)$\\risk-neutral \end{tabular}}} &
	  \textbf{2} 	& 0.0000 & 0.0072 & 0.0011 & 0.0059 & 0.31\\
	  &\textbf{3} 	& 0.0001 & 0.0104 & 0.0007 & 0.0051 & 0.40\\
	  &\textbf{5} 	& 0.0001 & 0.0194 & 0.0005 & 0.0053 & 0.46\\
	  &\textbf{10} & 0.0001 & 0.0303 & 0.0003 & 0.0047 & 0.73\\ \cline{1-7}
	 \multirow{4}{*}{\textbf{\begin{tabular}[c]{@{}c@{}}Uniform\\$\mathcal U(0,10)$\\risk-averse \end{tabular}}} & 
	 \textbf{2}		& 0.0003 & 0.0057 & 0.0012 & 0.0065 & 0.46\\
	& \textbf{3} 	& 0.0001 & 0.0069 & 0.0008 & 0.0048 & 0.52\\
	& \textbf{5}	& 0.0001 & 0.0161 & 0.0006 & 0.0066 & 0.63\\
	& \textbf{10} 	& 0.0002 & 0.0383 & 0.0005 & 0.0085 & 0.93\\\cline{1-7}
	 \multirow{4}{*}{\textbf{\begin{tabular}[c]{@{}c@{}}Gaussian\\ $\mathcal N(15,10^2)$ \\ risk-neutral \end{tabular}}} &
	 \textbf{2} 	& 0.0079 & 0.3684 & 0.0443 & 0.4394 & 0.31\\
	& \textbf{3} 	& 0.0103 & 0.4478 & 0.0225 & 0.9723 & 0.39\\
	& \textbf{5} 	& 0.0172 & 0.8819 & 0.0176 & 1.7324 & 0.45\\
	& \textbf{10} 	& 0.0169 & 1.8801 & 0.0118 & 2.1660 & 0.68\\ \bottomrule 
	
    \end{tabular}}

		\label{table:single-item_table} 
\end{table}

%% file: figures/auction_table_llg_new.tex
\begin{table}[htp]
	\caption{Results of NPGA in LLG-settings with independent and correlated valuations. Values are means of 10 runs of $5{,}000$ iterations. For FPSB, no BNE is known; for correlated priors, estimating $\hat \ell, \hat \epsilon$ is not straightforward\textsuperscript{4}. Negative $\ell^*$ are artefacts of the sample variance of $F_v$ at available precision.}\label{table:auction-results-llg}  
	\centering
    \footnotesize\begin{tabular}{lllrrrrrr}
    \multirow{2}{*}{\footnotesize{\textbf{priors}}} &
	\multirow{2}{*}{\footnotesize\textbf{\footnotesize\textbf{payment}}} & \multirow{2}{*}{\footnotesize\textbf{bidder}} &
    \multirow{2}{*}{$\ell^*$} & \multirow{2}{*}{\footnotesize{$\lVert\beta^*{-}\beta\rVert$}} & 
    \multirow{2}{*}{$\hat{\ell}$} &      \multirow{2}{*}{$\hat{\epsilon}$} & 
      {\multirow{2}{*}{\begin{tabular}[r]{@{}r@{}}\footnotesize\textbf{time}\\\footnotesize{sec/it} \end{tabular}}} \\ 
     \\
     \toprule
     \parbox[t]{2mm}{\multirow{8}{*}{\rotatebox[origin=c]{90}{\footnotesize\textbf{independent}}}}
	&\multirow{2}{*}{\footnotesize\textbf{\begin{tabular}[c]{@{}c@{}}n.-VCG \end{tabular}}}
	& \footnotesize\textbf{locals} & 0.0001 & 0.0050 & 0.0002 & 0.0009 &
	 \multirow{2}{*}{0.84} \\
	&& \footnotesize\textbf{global} & 0.0000 & 0.0269 & 0.0000 & 0.0001 \\	
		\cline{2-8}
	&\multirow{2}{*}{\footnotesize\textbf{\begin{tabular}[c]{@{}c@{}}n.-bid \end{tabular}}}
	& \footnotesize\textbf{locals} & -0.0002 & 0.0073 & 0.0003 & 0.0013 &
	 \multirow{2}{*}{0.79} \\
	&& \footnotesize\textbf{global} & 0.0000 & 0.0424 & 0.0000 & 0.0001  \\
		\cline{2-8}
	&\multirow{2}{*}{\footnotesize\textbf{\begin{tabular}[c]{@{}c@{}}n.-zero \end{tabular}}}
	&\footnotesize\textbf{locals} & -0.0001 & 0.0078 & 0.0002 & 0.0019 &
	  \multirow{2}{*}{0.79} \\
	&& \footnotesize\textbf{global} & 0.0000 & 0.0088 & 0.0000 & 0.0001 & \\ 
		\cline{2-8}
	&\multirow{2}{*}{\footnotesize\textbf{\begin{tabular}[c]{@{}c@{}}FPSB \end{tabular}}}
	& \footnotesize\textbf{locals} & -- & -- & 0.0009 & 0.0031 &
	   \multirow{2}{*}{0.65} \\
	&& \footnotesize\textbf{global} & -- & -- & 0.0016 & 0.0064 &\\ 
	\hline
	
	\parbox[t]{2mm}{\multirow{6}{*}{\rotatebox[origin=c]{90}{\footnotesize\textbf{correlated}}}}
	&\multirow{2}{*}{\footnotesize\textbf{\begin{tabular}[c]{@{}c@{}}n.-VCG \end{tabular}}}
	& \footnotesize\textbf{locals} & -0.0001 & 0.0042 & -- & -- & \multirow{2}{*}{0.80}\\
	&& \footnotesize\textbf{global} & 0.0000 & 0.0305 & -- & -- &\\
		\cline{2-8}
	&\multirow{2}{*}{\footnotesize\textbf{\begin{tabular}[c]{@{}c@{}}n.-bid \end{tabular}}}
	& \footnotesize\textbf{locals}& 0.0003 & 0.0064 &  -- & --   &\multirow{2}{*}{0.83}\\
	&& \footnotesize\textbf{global}& 0.0000 & 0.0498 &  -- & 	-- &\\	
		\cline{2-8}
	&\multirow{2}{*}{\footnotesize\textbf{\begin{tabular}[c]{@{}c@{}}n.-zero \end{tabular}}}
	& \footnotesize\textbf{locals} & 0.0001 & 0.0059 &  -- &  -- &\multirow{2}{*}{0.81}\\
	&& \footnotesize\textbf{global} & 0.0000 & 0.0072 &   -- & 	-- &\\    	
	\bottomrule
	\end{tabular}
\end{table}

%% file: figures/auction_table_llllgg_new.tex
\begin{table}
	\caption{Results and runtime of NPGA after 5{,}000 (1{,}000) iterations in the LLLLGG first-price (nearest-VCG) auction over 10 (2) repititions. Values are mean and (standard deviation).}\label{table:auction-results-llllgg}
	\centering
	\small\begin{tabular}{lrrrr}
		
		\multirow{2}{*}{\textbf{payment}} & \multirow{2}{*}{\textbf{bidder}} & \multirow{2}{*}{$\hat{\ell}$} & \multirow{2}{*}{$\hat{\epsilon}$}  & 
		\multicolumn{1}{l}{\multirow{2}{*}{\begin{tabular}[r]{@{}r@{}}\textbf{time} \\sec/iter\end{tabular}}} \\ 
		&& \\
		\toprule
		\multirow{2}{*}{\textbf{\begin{tabular}[l]{@{}l@{}}first-price \end{tabular}}} 
		& \textbf{locals} 	& 0{.}0015 (0{.}0003) & 0{.}0109 (0{.}0025)& 
		 \multirow{2}{*}{\begin{tabular}[r]{@{}r@{}}0{.}97\\(0{.}005) \end{tabular}} \\
		& \textbf{globals} 	& 0{.}0010 (0{.}0002) & 0{.}0077 (0{.}0016) \\
		\midrule		
		\multirow{2}{*}{\textbf{\begin{tabular}[l]{@{}l@{}}near.-VCG \end{tabular}}} 
		& \textbf{locals} 	& 0{.}0013 (0{.}0003) & 0{.}0052 (0{.}0012) & \multirow{2}{*}{\begin{tabular}[r]{@{}r@{}}275{.}22\\(0{.}670) \end{tabular}} \\
		& \textbf{globals} 	& 0{.}0011 (0{.}0006) & 0{.}0098 (0{.}0059) \\
	\end{tabular}
\end{table}